\documentclass[11pt]{article}

\usepackage{amsthm,amsmath,bm,amssymb,stmaryrd,fullpage,times}
\usepackage{tikz,url,program,cite}

\newcommand{\ket}[1]{\mid\hspace*{-4pt}{#1}\rangle}

\newtheorem{proposition}{{\bf Proposition}}[section]
\newtheorem{lemma}{{\bf Lemma}}[section]
\newtheorem{theorem}{{\bf Theorem}}[section]
\newtheorem{corollary}{{\bf Corollary}}[section]

\makeatletter
\def\mop#1{\mathop{\operator@font {#1\null}}}
\makeatother
\def\bigO{{\mop{O}}}
\def\hw{{\mop{hw}}}
\def\ord{{\mop{ord}}}
\def\Tr{{\mop{Tr}}}

\begin{document}
\title{A quantum circuit to find discrete logarithms on \hbox{ordinary binary elliptic curves} in~depth~$\bigO(\log^2n)$}

\author{Martin R{\"o}tteler\\
Microsoft Research\\
One Microsoft Way\\
Redmond, WA 98052\\
{\tt martinro@microsoft.com}
\and Rainer Steinwandt\\
Florida Atlantic University\\
Department of Mathematical Sciences\\
Boca Raton, FL 33431\\
{\tt rsteinwa@fau.edu}
}

\maketitle

\begin{abstract}
   Improving over an earlier construction by Kaye and Zalka \cite{KaZa04}, in \cite{MMCP09b} Maslov et al. describe an implementation of Shor's algorithm, which can solve the discrete logarithm problem on ordinary binary elliptic curves in quadratic depth $\bigO(n^2)$. In this paper we show that discrete logarithms on such curves can be found with a quantum circuit of depth $\bigO(\log^2n)$. As technical tools we introduce quantum circuits for ${\mathbb F}_{2^n}$-multiplication in depth $\bigO(\log n)$ and for ${\mathbb F}_{2^n}$-inversion in depth $\bigO(\log^2 n)$.
\end{abstract}

\section{Introduction}
The practical significance of the discrete logarithm problem on ordinary binary elliptic curves (see, e.\,g., \cite{FIPS1864}) renders these groups a natural target for implementing Shor's algorithm \cite{Sho97}. To implement an actual discrete logarithm computation, efficient quantum circuits to implement the pertinent curve arithmetic are needed, and a number of authors have explored circuits to implement the relevant elliptic curve operations \cite{KaZa04,MMCP09b,ARS12b}. When considering a complete implementation of Shor's algorithm in such a group, Maslov et al.'s proposal in \cite{MMCP09b} shows that a quadratic depth circuit is sufficient. The reason for the quadratic depth is essentially two-fold: a double-and-add computation to compute the relevant scalar multiplications in Shor's algorithm and a finite field inversion are the dominating operations. As shown in \cite{ARS12,ARS12b}, inversion in ${\mathbb F}_{2^n}$ can be implemented in depth $\bigO(n\log n)$, and so one may hope that the quadratic depth bound can indeed be overcome.

\paragraph{Our contribution.} Below we show that an appropriate organization of the scalar multiplication(s) in Shor's algorithm in combination with an improved ${\mathbb F}_{2^n}$-arithmetic enables a solution to the discrete logarithm problem on ordinary binary elliptic curves in depth $\bigO(\log^2n)$. To implement the necessary group operations we use complete binary Edwards curves as described in \cite{ARS12b}.

\paragraph{Structure of the paper.}
In the next section we briefly review some background on elliptic curves and Shor's algorithm. In particular we recall the the definition of binary Edwards curves as needed for the main part of the paper. Section~\ref{sec:parallel} details how with this curve representation the addition of any two curve points can be implemented in logarithmic depth.  Thereafter we discuss different options to organize the scalar multiplications in Shor's algorithm, including a tree-based approach with polylogarithmic depth. After addressing the technical point of deriving a unique representation of group elements, in Section~\ref{sec:theend} we establish our main result.

\section{Technical tools}
This section reviews some known results on ordinary binary elliptic curves and on computing discrete logarithms with Shor's algorithm.

\subsection{Quantum circuits for binary elliptic curve arithmetic}\label{sec:shortweierstrass}
For $n\in{\mathbb N}$ a positive integer, we denote by $\mathbb F_{2^n}$ a finite field of size $2^n$---for a cryptographic application, e.\,g., a digital signature scheme, a typical choice would be $n=163$ or $n=233$ \cite{FIPS1864}. To represent elements in $\mathbb F_{2^n}$ we use a polynomial basis representation. In other words we fix an irreducible polynomial $p\in{\mathbb F}_2[x]$ with coefficients in the integers modulo $2$ and identify $\mathbb F_{2^n}$ with the quotient ${\mathbb F}_2[x]/(p)$, so that each $\alpha\in{\mathbb F}_{2^n}$ has a unique expression of the form $\alpha=\sum_{i=0}^{n-1}\alpha_i\cdot(x^i+(p))$ with $(\alpha_0,\dots,\alpha_{n-1})\in{\mathbb F}_2^n$.

Using a \emph{short Weierstra{ss} form}, each ordinary binary elliptic curve can be represented by a polynomial equation
\begin{equation}
   y^2+xy=x^3+a_2x^2+a_6,\label{equ:weierstrass}
\end{equation}
where $a_2,a_6\in {\mathbb F}_{2^n}$ with $a_6\ne 0$ \cite[Chapters~13.1.4 and 13.1.5]{CoFr06}. More precisely, the elliptic curve represented by Equation~\eqref{equ:weierstrass} consists of the points
${\mathrm E}_{a_2,a_6}({\mathbb F}_{2^n}):=\{(u,v)\in{\mathbb F}_{2^n}: v^2+uv=u^3+a_2u^2+a_6\}\cup\{{\mathcal O}\}$,
where ${\mathcal O}\in {\mathrm E}_{a_2,a_6}({\mathbb F}_{2^n})$ is the unique projective point that is obtained when homogenizing Equation~\eqref{equ:weierstrass}. 
Hasse's bound implies that the size of ${\mathrm E}_{a_2,a_6}({\mathbb F}_{2^n})$ differs from $2^n$ by no more than $2^{1+(n/2)}+1$, and the subgroups of ${\mathrm E}_{a_2,a_6}({\mathbb F}_{2^n})$ considered in cryptographic applications typically have a very small cofactor. Hence $n$ is a natural parameter to measure the complexity of a quantum circuit to solve the discrete logarithm problem in a ${\mathrm E}_{a_2,a_6}({\mathbb F}_{2^n})$.

The set ${\mathrm E}_{a_2,a_6}({\mathbb F}_{2^n})$ has a group structure, but implementing this group law directly comes at a certain inconvenience: case distinctions have to be made, which require the implementation of a (nested) if-then-else statement (cf. the discussion in \cite{Sol98,ARS12b}). To avoid this issue, subsequently we use \emph{complete binary Edwards curves} as introduced by Bernstein et al. \cite{BLF08}. For $n\ge 3$ each ordinary elliptic curve is birationally equivalent to a complete binary Edwards curve, and we can represent such a curve by an equation 
\begin{equation}d_1(x+y)+d_2(x^2+y^2)=xy+xy(x+y)+x^2y^2\label{equ:edwards}
\end{equation}
with $d_1\in{\mathbb F}_{2^n}^*$ being non-zero, $d_2\in{\mathbb F}_{2^n}$ and  $\Tr(d_2)=1$.\footnote{The condition $\Tr(d_2)=1$ can equivalently be expressed as $\sum_{i=0}^{n-1}d_2^{2^i}=1$.} By ${\mathrm E}_{{\mathrm B},d_1, d_2}({\mathbb F}_{2^n})$ we denote the points in ${\mathbb F}_{2^n}^2$ satisfying Equation~\eqref{equ:edwards}. The group law on ${\mathrm E}_{{\mathrm B},d_1, d_2}({\mathbb F}_{2^n})$ is given by the formula
\begin{eqnarray*}
   x_3&=&\frac{d_1(x_1+x_2)+d_2(x_1+y_1)(x_2+y_2)+(x_1+x_1^2)(x_2(y_1+y_2+1)+y_1y_2)}{d_1+(x_1+x_1^2)(x_2+y_2)}\text{ and}\\
   y_3&=&\frac{d_1(y_1+y_2)+d_2(x_1+y_1)(x_2+y_2)+(y_1+y_1^2)(y_2(x_1+x_2+1)+x_1x_2)}{d_1+(y_1+y_1^2)(x_2+y_2)},
\end{eqnarray*}
where $(x_1,y_1),(x_2,y_2)\in {\mathrm E}_{{\mathrm B},d_1, d_2}({\mathbb F}_{2^n})$ can be arbitrary curve points---including the identity element $(0,0)$. To derive efficient addition formula one can (similarly as for a Weierstrass form) pass to projective coordinates. Bernstein et al. \cite{BLF08} show that from projective representations $(X_1, Y_1, Z_1)$, $(X_2, Y_2, Z_2)$ of two points one can derive a projective representation $(X_3, Y_3, Z_3)$ of their sum by means of 21 multiplications in ${\mathbb F}_{2^n}$, four multiplications by one of the constant $d_1$, $d_2$, one squaring and 15 additions in ${\mathbb F}_{2^n}$.
$$\addtolength{\arraycolsep}{-0.75pt}
\begin{array}{l}
\begin{array}{llllllllllll}
  W_1 &=& X_1 + Y_1,  & W_2 &=& X_2 + Y_2,& A &=& X_1\cdot (X_1 + Z_1),& B& =& Y_1 \cdot (Y_1 + Z_1),\\
  C   &=&Z_1\cdot Z_2,& D  &=& W_2\cdot Z_2,& E&=& d_1C^2,& H &=& (d_1Z_2 + d_2W_2)\cdot W_1 \cdot C,\\
  I &=& d_1Z_1\cdot C,& U &=& E + A\cdot D,& V &=& E + B\cdot D,& S &=& U\cdot V,\\\hline
\end{array}\\
\begin{array}{lll}
   X_3 &=& S\cdot Y_1 + (H + X_2\cdot(I + A\cdot(Y_2 + Z_2)))\cdot V\cdot Z_1,\rule{0ex}{2.3Ex}\\
   Y_3 &=& S\cdot X_1 + (H + Y_2\cdot(I + B\cdot(X_2 + Z_2)))\cdot U\cdot Z_1,\\
   Z_3 &=& S\cdot Z_1.
\end{array}
\end{array}$$
From an asymptotic point of view, it suffices to observe that the number of field operations is constant. It is not necessary, however, to perform these field operations sequentially, and \cite{ARS12b} suggest some parallelization, establishing the following upper bound for the depth of a point addition circuit, where $D_M(n)$ stands for the depth of an ${\mathbb F}_{2^n}$-multiplier $\ket{\alpha}\ket{\beta}\ket{\gamma}\mapsto\ket{\alpha}\ket{\beta}\ket{\gamma+\alpha\beta}$.

\begin{proposition}[{\cite[Proposition~3.3]{ARS12b}}]\label{prop:edadd} Let $(X_1,Y_1,Z_1)$ and $(X_2,Y_2,Z_2)$ be projective representations of two (not necessarily different) points $P_1, P_2\in{\mathrm E}_{{\mathrm B},d_1, d_2}({\mathbb F}_{2^n})$. Then the addition map
$$\ket{X_1}\ket{Y_1}\ket{Z_1}\ket{X_2}\ket{Y_2}\ket{Z_2}\ket{0}\ket{0}\ket{0}\longrightarrow\ket{X_1}\ket{Y_1}\ket{Z_1}\ket{X_2}\ket{Y_2}\ket{Z_2}\ket{X_3}\ket{Y_3}\ket{Z_3},$$
where $(X_3,Y_3,Z_3)$ is a projective representation of $P_1+P_2$, can be implemented in depth $5\cdot D_M(n)+4\cdot \max(D_M(n),2n)+\bigO(1)$. 
\end{proposition}

From \cite{MMCP09b} it follows that we can choose $D_M(n)\in\bigO(n)$, and in Section~\ref{sec:fieldmult} we will show that through a suitable use of trees $D_M(n)$ can be chosen to be of logarithmic depth. As the number of field operations to add to curve points is constant, this establishes immediately the existence of a logarithmic depth circuit for point addition. To optimize the circuit depth, we can exploit the bound from Proposition~\ref{prop:edadd}: looking into the proof of \cite[Proposition~3.3]{ARS12b}, one recognizes that the term $2n$ occurring as argument of the maximum in Proposition~\ref{prop:edadd} describes  the multiplication of a binary $n\times n$-matrix with a binary vector. In the next section we will see that such a multiplication can be realized in logarithmic depth as well.

A technical issue that we address in Section~\ref{sec:lowdepthdivision} is the derivation of the unique (affine) representation from a projective representation of a curve point: The natural way to realize this is by means of an inversion in ${\mathbb F}_{2^n}$, but for none of the division circuits described in \cite{KaZa04,MMCP09b,ARS12b} a polylogarithmic depth bound is available. We modify the construction in \cite{ARS12b} to achieve polylogarithmic depth.

\subsection{Shor's algorithm}
For our discussion we assume that a generator $P\in{\mathrm E}_{{\mathrm B},d_1, d_2}({\mathbb F}_{2^n})$ of a cyclic subgroup of ${\mathrm E}_{{\mathrm B},d_1, d_2}({\mathbb F}_{2^n})$ is fixed and the order $\ord(P)$ of this group generator is known. Moreover, we assume that a group element $Q$ in the subgroup generated by $P$ is fixed; our goal is to find the unique integer $r\in\{1,\dots,\ord(P)\}$ such that $Q=r\cdot P$. 
The algorithm proceeds as follows. First, two registers of length $n+1$ qubits\footnote{Hasse's bound guarantees that $\ord(P)$ can be represented with $n+1$ bits.} are created and each qubit is initialized in the $\ket{0}$ state. Then a Hadamard transform $H$ is applied to each qubit, resulting in the state $\frac{1}{2^{n+1}} \sum_{k,\ell=0}^{2^{n+1}-1} \ket{k,\ell}$. Next, conditioned on the content of the register holding the label $k$ or $\ell$, we add the corresponding multiple of $P$ and $Q$, respectively, i.\,e., we 
implement the map 
\[ 
\frac{1}{2^{n+1}} \sum_{k,\ell=0}^{2^{n+1}-1} \ket{k,\ell} \mapsto 
\frac{1}{2^{n+1}} \sum_{k,\ell=0}^{2^{n+1}-1} \ket{k,\ell}\ket{k P + \ell Q}. 
\]
Hereafter, the third register is discarded and a quantum Fourier transform ${\rm QFT}_{2^{2\cdot(n+1)}}$ on $2\cdot(n+1)$ qubits is computed. Finally, the state of the first two registers---which hold a total of $2\cdot(n+1)$ qubits---is measured. As shown in \cite{Sho97,Kitaev:97}, the factor $r$ can be computed from this measurement data via classical post-processing. The corresponding quantum circuit is 
shown in Figure \ref{fig:shorcircuit}. In the following sections, we will be concerned with parallelizing the parts of the circuit in this figure. In Section~\ref{sec:lowdepthdivision} we will address the problem of having a non-unique representation of curve points, as the above description of Shor's algorithm implicitly assumes group elements to have a unique representation.

\begin{figure}[hbt]
\centerline{
\includegraphics[width=0.9\textwidth]{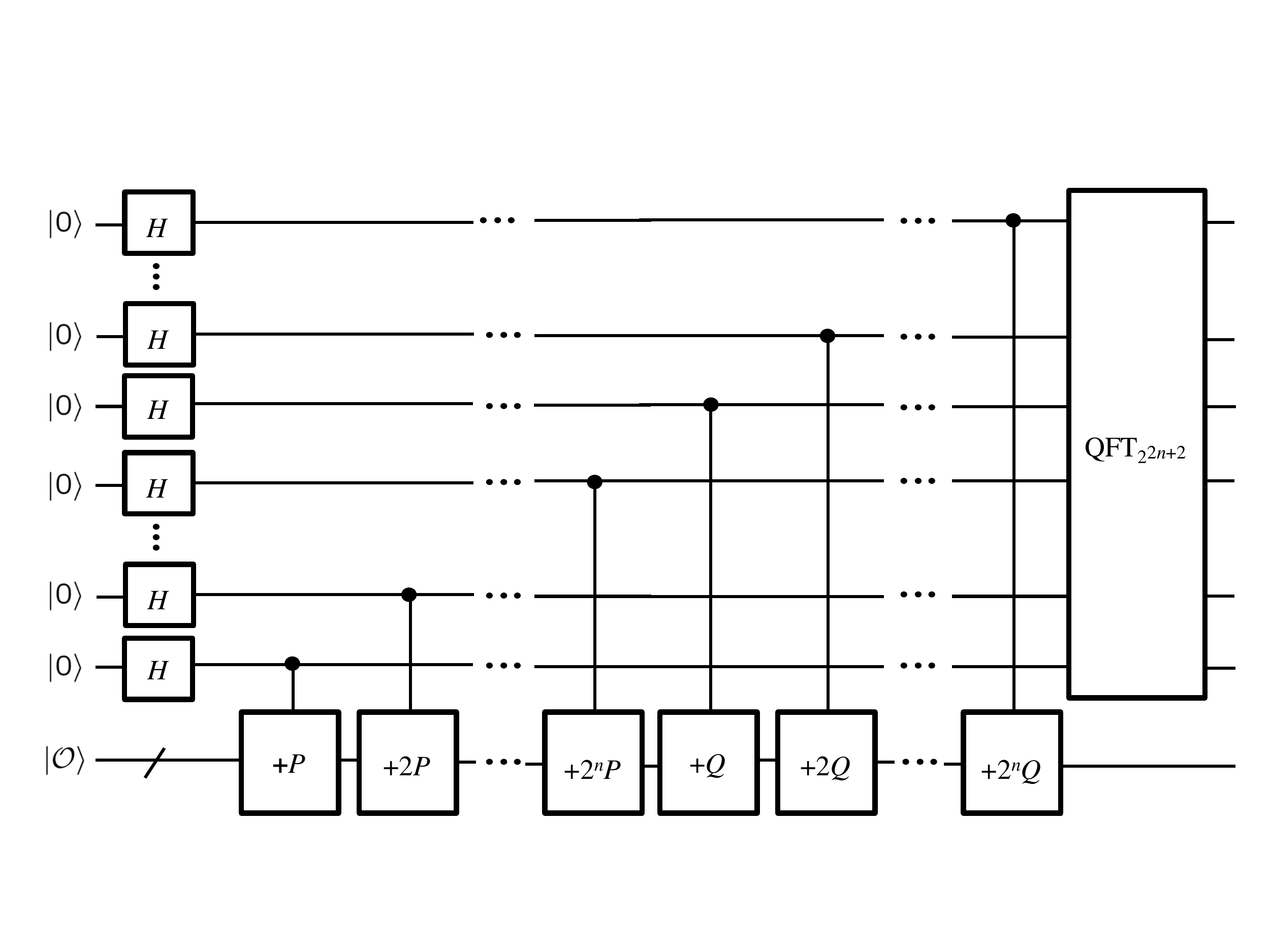}
}\vspace*{-9Ex}
\caption{\label{fig:shorcircuit} Shor's algorithm to compute the discrete logarithm in the subgroup of an elliptic curve generated by a point $P$. The input to the problem is a point $Q$, and the task is to find $r\in\{1,\dots,\ord(P)\}$ such that $Q = r \cdot P$. The circuit naturally decomposes into three parts, namely (i) the Hadamard layer on the left, (ii) a double scalar multiplication (in this figure implemented as a cascade of point additions), and (iii) the quantum Fourier transform ${\rm QFT}$ at the end. We show that each of these parts can be implemented in a circuit depth of $\bigO(\log^2n)$ to obtain the main result of this paper.}
\end{figure}

\section{Parallelizing Shor's algorithm}\label{sec:parallel}
To reduce the circuit depth, we parallelize Shor's algorithm on two different levels: (i) the computation of ${\mathbb F}_{2^n}$-multiplications is parallelized and (ii) the computation of the scalar products $k\cdot P+\ell\cdot Q$ is parallelized.

\subsection{Multiplying ${\mathbb F}_{2^n}$-elements in depth $\bigO(\log n)$}\label{sec:fieldmult}
A simple observation that will be useful is that we can implement the map $$\ket{x}\underbrace{\ket{0}\dots\ket{0}}_m\longmapsto\ket{x}\underbrace{\ket{x}\dots\ket{x}}_m\quad (x\in\{0,1\})$$ in depth $\bigO(\log m)$ by arranging $m$ CNOT gates as a tree. Figure~\ref{fig:multifanout}, which derives from \cite[Figure~2]{GHM02}, shows such a `multi-fan-out CNOT with $\ket{0}$-input' for the case $m=7$. We note that in general such a tree is not functionally equivalent to a CNOT with fan-out greater than $1$, but for the case of a $\ket{0}$-input this equivalence holds, and for our purposes this is the only case needed.

\begin{figure}[htb]
  \begin{center}
     \includegraphics[scale=0.7, trim = 0mm 200mm 140mm 20mm]{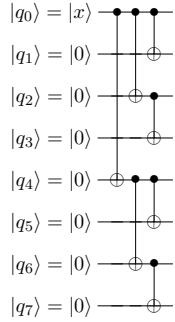}
  \end{center}
\caption{obtaining $m$ `copies' of a qubit in depth $\bigO(\log m)$}\label{fig:multifanout}
\end{figure}

As starting point to implement multiplication in ${\mathbb F}_{2^n}$ we use the circuit proposed by Maslov et al. in \cite{MMCP09b}---which builds on a classical Mastrovito multiplier \cite{Mas88,Mas91,MaHa04}. This construction reduces the task of multiplying two elements in ${\mathbb F}_2[x]/(p)$ to implementing a quantum circuit that evaluates two matrix-vector multiplications with a Toeplitz matrix, one matrix-vector multiplication with a matrix that depends only on the polynomial $p$, and an addition in ${\mathbb F}_2^n$. More specifically, the coefficients $(\gamma_0,\dots,\gamma_{n-1})$ of the product  $(\sum_{i=0}^{n-1}{\alpha_i}\cdot(x^i+(p)))\cdot(\sum_{i=0}^{n-1}{\beta_i}\cdot(x^i+(p)))$ are obtained as follows, where $M\in{\mathbb F}_2^{n\times (n-1)}$ is independent of the specific field elements to be multiplied; the matrix $M$ depends only on the irreducible polynomial $p$ defining the underlying finite field ${\mathbb F}_2[x]/(p)$:
\begin{equation}{\vec\gamma}=
\underbrace{\left(\begin{array}{ccccc}
\alpha_0 & 0 & \dots & 0 & 0\\
\alpha_1 &\alpha_0 &\dots & 0 & 0\\
\vdots&\vdots&\ddots&\vdots&\vdots\\
\alpha_{n-2}&\alpha_{n-3}&\dots&\alpha_0&0\\
\alpha_{n-1}&\alpha_{n-2}&\dots&\alpha_1&\alpha_0
\end{array}\right)}_{=L}\cdot{\vec\beta}+M\cdot
\underbrace{\left(\begin{array}{cccccc}
0&\alpha_{n-1}&\alpha_{n-2}&\dots&\alpha_2&\alpha_1\\
0&0&\alpha_{n-1}&\dots&\alpha_3&\alpha_2\\
\vdots&\vdots&\vdots&\ddots&\vdots&\vdots\\
0&0&0&\dots&\alpha_{n-1}&\alpha_{n-2}\\
0&0&0&\dots&0&\alpha_{n-1}
\end{array}\right)}_{=U}\cdot{\vec\beta}\label{equ:multformula}
\end{equation}

\subsubsection{Computing the products $L\cdot\vec\beta$ and $U\cdot\vec\beta$}
To implement the multiplications of $L$ and $U$ with $\vec\beta$ respectively, we first observe that---considering both the computation of  $L\cdot\vec\beta$ and $U\cdot\vec\beta$ combined---each coefficient $\alpha_i$ ($i=0,\dots,n-1$) occurs in exactly $n$ products of the form $\alpha_i\beta_j$. Similarly, each coefficient $\beta_j$ ($j=0,\dots,n-1$) occurs in a total of exactly $n$ products. We want to compute \emph{all} of these ${\mathbb F}_2$-products in parallel. So we ensure that $n$ `copies' of each of $\beta_0,\dots,\beta_{n-1}$ are available, using a `multi-fan-out CNOT with $\ket{0}$-input' for each $\beta_i$. As the CNOT trees for $\beta_i$ and $\beta_j$ with $i\ne j$ operate on disjoint wires, they can be executed in parallel and implemented in depth $\bigO(\log n)$. Analogously, using a `multi-fan-out CNOT with $\ket{0}$-input' for each of $\alpha_0,\dots,\alpha_{n-1}$ we can---in depth $\bigO(\log n)$ and in parallel to the trees for copying the $\beta_i$-values---provide $n$ `copies' of each of $\alpha_0,\dots,\alpha_{n-1}$.

Having, at the cost of $\bigO(n^2)$ qubits, all these copies at our disposal, we can now, in depth $1$, compute all products $\alpha_i\cdot\beta_j$ that are necessary to find $L\cdot\vec\beta$ and $U\cdot\vec\beta$ in parallel, using $n^2$ Toffoli gates. Having evaluated all these products we can simply compute each entry of $L\cdot\vec\beta$ and $U\cdot\vec\beta$ by using a a depth $\bigO(\log n)$ addition tree for each entry of these two vectors. Only CNOT gates (and no further ancillae) are needed for this. Figure~\ref{fig:UBandLB} shows an example for the case $n=2$, i.\,e., we have $$L=\left(\begin{array}{cc}\alpha_0&0\\\alpha_1&\alpha_0\end{array}\right), U=\left(\begin{array}{cc}0&\alpha_1\end{array}\right).$$In this case all occurring multiplications can be evaluated in depth $1=\log_2(2)$, and the final addition trees reduce to a single CNOT gate to compute the `last' entry $\alpha_1\beta_0+\alpha_0\beta_1$ of $L\cdot\vec\beta$. 
\begin{figure}[thb]
 \begin{center}
     \includegraphics[scale=0.7, trim = 0mm 167mm 140mm 20mm]{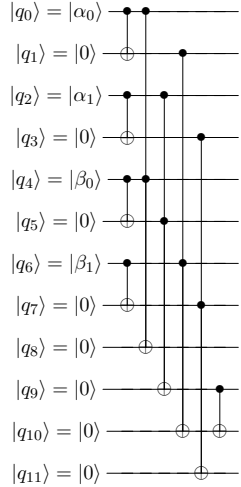}
  \end{center}
\caption{computing $L\cdot \vec\beta$ and $U\cdot \vec\beta$ for $n=2$: all four multiplications occur in parallel}\label{fig:UBandLB}
\end{figure}

\subsubsection{Multiplication by the constant matrix $M$}
From Equation~\eqref{equ:multformula}, we see that the vector $\vec\eta=U\cdot\vec\beta\in{\mathbb F}_{2}^{n-1}$ needs to be multiplied from the left with the fixed $n\times(n-1)$-matrix $M$. Writing $\hw(M_i)$ for the Hamming weight of the $i^\text{th}$ column of $M$ and denoting by $\eta_i$ the $i^\text{th}$ entry of $\vec\eta$, we first create $\hw(M_i)\le n$ copies of $\eta_i$ ($i=1,\dots,n-1$), requiring $\bigO(n^2)$ qubits. For this we use again `multi-fan-out CNOT gates with $\ket{0}$-input' that operate in parallel and can be realized in depth $\bigO(\log n)$. This allows us to compute all $n$ entries of $M\cdot\vec\eta$ in parallel: for each entry of the result we can use an $\bigO(\log n)$-depth addition tree that computes the scalar product of the corresponding row of $M$ with $\vec\eta$. As the matrix $M$ is fixed, this can be done by means of CNOT gates. Figure~\ref{fig:Mmult} shows a `worst case tree' of depth $3=\log_2(8)$ for the case $n-1=8$: multiplying a matrix row consisting entirely of $1$s with $(\eta_0,\dots,\eta_7)$.
\begin{figure}[hbt]
 \begin{center}
     \includegraphics[scale=0.7, trim = 0mm 200mm 140mm 20mm]{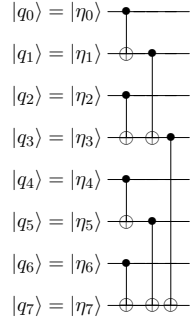}
  \end{center}
\caption{computing $(1,1,1,1,1,1,1,1)\cdot(\eta_0,\eta_1,\eta_2,\eta_3,\eta_4,\eta_5,\eta_6,\eta_7)^{\mathrm t}$ in depth $3$}\label{fig:Mmult}
\end{figure}

Finally, to complete the evaluation of Equation~\eqref{equ:multformula}, we add the binary vector $L\cdot\vec \beta$ to $M\cdot (U\cdot \vec B)$ by means of $n$ CNOT gates that operate in parallel. With all involved steps---computing $L\cdot\vec\beta$ and $U\cdot\vec\beta$, finding $M\cdot U\vec\beta$, and determining $L\vec\beta+MU\vec\beta$---being realizable in depth $\bigO(\log n)$ we obtain the following result.
\begin{theorem}[${\mathbb F}_{2^n}$-multiplication in logarithmic depth]\label{prop:GF2nlowdepthmult} There is a polynomial-size quantum circuit of depth $\bigO(\log n)$ which on input polynomial basis representations of $\alpha,\beta\in {\mathbb F}_{2^n}$ computes a polynomial basis representation of the product $\alpha\cdot\beta\in{\mathbb F}_{2^n}$.
 \end{theorem}
The above-described technique to multiply a vector with the fixed matrix $M$ can also be used to implement other matrix-vector multiplications: Given a binary $n\times n$ matrix $A\in{\mathbb F}_2^{n\times n}$ and a vector $\vec \gamma\in{\mathbb F}_2^n$,  we first use $n$ `multi-fan-out CNOT with $\ket{0}$-input' to create $n$ copies of $\vec\gamma$, investing $\bigO(n^2)$ qubits. Handling all entries of $\vec\gamma$ in parallel, this can be done in depth $\bigO(\log n)$. Hereafter we can use an addition tree to compute the necessary $n$ scalar products in parallel, just as in the discussion of the matrix $M$. We can apply this observation to the matrix multiplications occurring in the proof of Proposition~\ref{prop:edadd} given in \cite{ARS12b}, which replaces the argument $2n$ of $\max$ with a function in $\bigO(\log n)$. In combination with Theorem~\ref{prop:GF2nlowdepthmult}, we obtain the following.
\begin{corollary}\label{cor:lowdepthadd}
   Using projective coordinates, a projective representation of the sum of two points on a complete binary Edwards curve ${\mathrm E}_{{\mathrm B},d_1, d_2}({\mathbb F}_{2^n})$ can be computed in depth $\bigO(\log n)$.
\end{corollary}

\subsection{Organizing the computation of $k\cdot P+\ell\cdot Q$}\label{sec:doubleandadd101}
An essential part of an implementation of Shor's algorithm is a circuit which on input (binary representations of) $k=\sum_{i=0}^n k_i2^i$ and $\ell=\sum_{i=0}^n\ell_i2^i$ computes a (unique) representation of $k\cdot P+\ell\cdot Q$---because of Hasse's bound, we can assume that $k$ and $\ell$ are represented with (at most) $n+1$ qubits each.

\subsubsection{Sequential double-and-add}\label{sec:seq2a}
The approach taken in \cite{MMCP09,MMCP09b} can be seen as implementation of a right-to-left version of the double-and add-algorithm\footnote{More specifically, Maslov et al. first perform all necessary additions of the points $2^i\cdot P$ and then continue with the necessary additions of the points $2^i\cdot Q$; this change of addition order does not affect the circuit depth.}:

\begin{center}
\begin{minipage}{0.6\textwidth}
\begin{program}
R\gets{\mathcal O}\rcomment{\# initialize result to identity}
\FOR i=0\TO n \STEP 1
  \IF{k_i=1}\THEN R\gets R+2^i\cdot P
  \untab\untab\IF{\ell_i=1}\THEN R\gets R+2^i\cdot Q
\untab\untab\untab|return|\ R\rcomment{\# $R=kP+\ell Q$}
\end{program}
\end{minipage}
\end{center}

Applying Corollary~\ref{cor:lowdepthadd}, each point addition requires circuit depth $\bigO(\log n)$, and so we immediately obtain an $\bigO(n\log n)$-depth circuit to compute $kP+\ell Q$. 
A feature of this particular strategy is that all points $2^iP$ and $2^iQ$ can be precomputed classically, and in all adders involved one argument is constant. This can be exploited to simplify the addition circuits, but for a realistic value of $n$ several hundred different addition circuits are involved. A left-to-right formulation of the double-and-add algorithm offers an alternative that involves only three types of circuits, but this uniformity comes at the cost of a general doubling circuit:

\begin{center}
\begin{minipage}{0.6\textwidth}
\begin{program}
R\gets{\mathcal O}\rcomment{\# initialize result to identity}
\IF{k_{n}=1}\THEN R\gets R+P\rcomment{\quad\# adjust starting value based on most significant bit}
\untab\untab\IF{\ell_{n}=1}\THEN R\gets R+Q
\untab\untab\FOR i=n-1\TO 0\STEP -1
R\gets 2\cdot R
\IF {k_i=1} \THEN R\gets R+P
\untab\untab\IF{\ell_i=1}\THEN R\gets R+Q
\untab\untab\untab|return|\ R\rcomment{\# $R=kP+\ell Q$}
\end{program}
\end{minipage}
\end{center}
This algorithm processes $k$ and $\ell$ simultaneously, so that the total number of doublings is $n$ (rather than $2n$)---the latter technique goes back to Straus \cite{Str64} and is also known as \emph{Shamir's trick} (cf. \cite{ElG84}). Realizing point doubling and addition of $P$ respectively $Q$ in depth $\bigO(\log n)$, this yields again an $\bigO(n\log n)$-depth circuit to find $kP+\ell Q$.

\subsubsection{Parallelized double-and-add}\label{sec:doubleAdd}
To find $kP+\ell Q$ with a quantum circuit of smaller depth, we parallelize the computation of $kP+\ell Q$. To simplify the description we assume that the length $n+1$ of the binary representation of $k$ and of $\ell$ is a power of $2$---if this is not the case, the binary expansions can be padded with $0$s accordingly.

\begin{center}
\begin{minipage}{0.9\textwidth}
\begin{program}
  (R_{0}^{(\log_2(n+1))},\dots,R_{n}^{(\log_2(n+1))})\gets(k_0\cdot 2^0P,\dots,k_n\cdot 2^nP)
  (S_{0}^{(\log_2(n+1))},\dots,S_n^{(\log_2(n+1)}))\gets(\ell_0\cdot 2^0S,\dots,\ell_n\cdot 2^nS)
   \FOR i=\log_2(n+1)-1 \TO 0 \STEP -1
      (R^{(i)}_{0},\dots,R^{(i)}_{j},\dots,R^{(i)}_{2^i-1})=(R^{(i+1)}_0+R_1^{(i+1)},\dots,R^{(i+1)}_{2j}+R^{(i+1)}_{2j+1},\dots, R^{(i+1)}_{2^{i+1}-2}+R^{(i+1)}_{2^{i+1}-1})
       (S^{(i)}_{0},\dots,S^{(i)}_{j},\dots,S^{(i)}_{2^i-1})=(S^{(i+1)}_0+S_1^{(i+1)},\dots,S^{(i+1)}_{2j}+S^{(i+1)}_{2j+1},\dots, S^{(i+1)}_{2^{i+1}-2}+S^{(i+1)}_{2^{i+1}-1})\untab 
   |return|\ R_{0}^{(0)}+S_{0}^{(0)}
\end{program}
\end{minipage}
\end{center}
The summation computed here is essentially the same as in the first sequential version of the double-and-add algorithm we discussed, therewith following the approach in \cite{MMCP09,MMCP09b}. Like Maslov et al. we parallelize the execution of gates that operate on disjoint wires, but we process more than one point $2^iP$ respectively $2^jQ$ at a time. The (affine) points $2^0P,\dots, 2^nP$ and $2^0Q,\dots,2^nQ$ can be precomputed classically, and the leaves can be initialized with a sequence of CNOT gates, conditioned on the individual bits of the binary representation of $k$ and $\ell$. Using `multi-fan-out CNOTs with $\ket{0}$-input', we can create $2n$ copies of the binary representation of $k$ and $2n$ copies of the binary representation of $\ell$ in depth $\bigO(\log n)$. Each leaf corresponds to a separate quantum register of linear length, initialized with $\ket{\mathcal O}=\ket{(0,0,1)}$, and hence the CNOT gates for different points operate on disjoint wires. Having $2n$ copies of each coefficient $k_i$ and $\ell_j$ available, we can copy all $2n$ qubits representing a point $2^iP$ or $2^jQ$ on ${\mathrm E}_{{\mathrm B},d_1, d_2}({\mathbb F}_{2^n})$ in constant depth.\footnote{The `projective coordinate' of $2^iP$ and $2^jQ$ is always $1$.} In particular, the complete initialization can be realized in depth $\bigO(\log n)$.

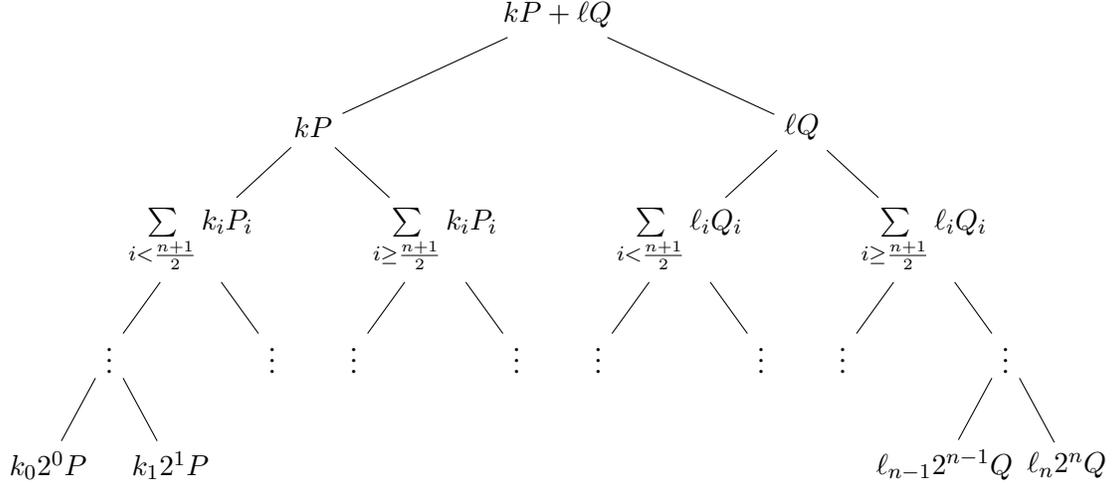
\begin{figure}
\begin{center}
\begin{tikzpicture}[level/.style={sibling distance=65mm/#1}]
\node {$kP+\ell Q$}
  child {node  {$kP$}
    child {node  {$\sum\limits_{i<\frac{n+1}{2}}k_iP_i$}
      child {node {$\vdots$}
        child {node   {$k_0{2^0}P$}}
        child {node   {$k_12^1P$}}
      } 
      child {node {$\vdots$}}
    }
    child {node  {$\sum\limits_{i\ge\frac{n+1}{2}}k_iP_i$}
      child {node {$\vdots$}}
      child {node {$\vdots$}}
    }
  }
  child {node  {$\ell Q$}
    child {node {$\sum\limits_{i<\frac{n+1}{2}}\ell_iQ_i$}
      child {node {$\vdots$}}
      child {node {$\vdots$}}
    }
  child {node  {$\sum\limits_{i\ge\frac{n+1}{2}}\ell_iQ_i$}
    child {node {$\vdots$}}
    child {node {$\vdots$}
      child {node  {$\ell_{n-1}2^{n-1}Q$}}
      child {node {$\ell_n2^nQ$}}}}};
\end{tikzpicture}
\caption{computing $kP+\ell Q$: parallelized double-and-add}\label{fig:treeaddition}
\end{center}
\end{figure}

 To proceed from one tree level to the next, a layer of addition circuits is used that operate in parallel, each circuit adding two curve points. Realizing each such addition circuit in logarithmic depth as in Corollary~\ref{cor:lowdepthadd} results in a quantum circuit of depth $\bigO(\log^2 n)$ to compute $kP+\ell Q$. Figure~\ref{fig:treeaddition} shows the resulting tree structure to compute $kP+\ell Q$. We obtain the following lemma.

\begin{lemma}
Using projective coordinates, for a complete binary Edwards curve the map $$\ket{k}\ket{\ell}\ket{0}\mapsto\ket{k}\ket{\ell}\ket{kP+\ell Q},$$ where $0\le k,\ell<2^{n+1}$, can be realized in depth $\bigO(\log^2 n)$.
\end{lemma}

%
%

\subsubsection{Deriving a unique point representation}\label{sec:lowdepthdivision}
In our discussion of the elliptic curve arithmetic we focused on the use of projective coordinates, as this avoids the use of a (costly) ${\mathbb F}_{2^n}$-inversion. We pay for this by a non-unique point representation, however. To pass from projective coordinates $(X_1,Y_1,Z_1)\in{\mathbb F}_{2^n}^3$ for a point on ${\mathrm E}_{{\mathrm B},d_1, d_2}({\mathbb F}_{2^n})$ to the unique affine coordinates $(X_1/Z_1,Y_1/Z_1)$, it is sufficient to compute the inverse of $Z_1\in{\mathbb F}_2^*$ followed by two multiplications in ${\mathbb F}_{2^n}$. From Theorem~\ref{prop:GF2nlowdepthmult} we know that these two multiplications can be implemented in depth $\bigO(\log n)$, and we are left with the task of computing $Z_1^{-1}$. In \cite{ARS12b} a depth $\bigO(n\log n)$ circuit is proposed for this task which builds on a classical algorithm by Itoh ans Tsuji \cite{ItTs89}. This algorithm consists of two main parts followed by a final squaring in ${\mathbb F}_{2^n}$.

The two main parts require a total of $\bigO(\log n)$ multiplications in ${\mathbb F}_{2^n}$ and $\bigO(\log n)$ exponentiations with fixed powers of $2$. The latter maps are bijective and ${\mathbb F}_2$-linear and can be implemented as matrix-vector multiplications with a fixed matrix. Consequently all necessary computations in these two parts can be implemented in depth $\bigO(\log^2 n)$. The final squaring operation can again be realized as a matrix-vector multiplication with a fixed matrix, and so the complete inversion can be implemented in depth $\bigO(\log^2 n)$.
\begin{theorem}
There is a polynomial-size quantum circuit of depth $\bigO(\log^2 n)$ which on input a polynomial basis representation of $\alpha\in {\mathbb F}_{2^n}^*$ computes a polynomial basis representation of the inverse $\alpha^{-1}\in{\mathbb F}_{2^n}^*$. 
\end{theorem}
Implementing an ${\mathbb F}_{2^n}$-inverter in this way and combining it with two multiplication circuits as in Theorem~\ref{prop:GF2nlowdepthmult}, we can derive the unique affine representation of a curve point in polylogarithmic depth. More specifically, we have the following.
\begin{corollary}
There is a polynomial-size quantum circuit of depth $\bigO(\log^2 n)$ which on input a projective representation $(X_1,Y_1,Z_1)$ of a point on ${\mathrm E}_{{\mathrm B},d_1, d_2}({\mathbb F}_{2^n})$, returns the affine representation $(X_1/Z_1, Y_1/Z_1)$ of this point.
\end{corollary}
\section{Computation of a discrete logarithm in depth $\bigO(\log^2n)$}\label{sec:theend}

We now have all the pieces in place to establish our main result. 

\begin{corollary}
Shor's algorithm to compute the discrete logarithm on a complete binary Edwards curve can be implemented using a quantum circuit of polynomial size and depth $\bigO(\log^2n)$. 
\end{corollary}
\begin{proof}
Referring to the three stages as in Figure~\ref{fig:shorcircuit}, we first note that the Hadamard gates on the left can be implemented in depth $1$. Next, as shown in Subsection \ref{sec:doubleAdd}, we can implement the operation that computes $k P + \ell Q$ in depth $\bigO(\log^2n)$ with a circuit of size polynomial in $n$. 
In order to complete the proof, we have to bound the depth of the Fourier transform ${\rm QFT}_{2^{2n+2}}$. For this, we use a result shown in \cite{CW:2000} that the Fourier
transform on $m$ qubits with target error $\varepsilon$ can be implemented with a circuit of depth $\bigO(\log{m}+\log\log{1/\varepsilon})$ and polynomial size. Choosing $m=n+1$ and $\varepsilon=1/2^{2m}$ is sufficient for Shor's algorithm to find the discrete logarithm  
with constant probability of success \cite{Sho97}, i.\,e., we can upper bound the depth of the ${\rm QFT}$ by $\bigO(\log{n})$ and hence the overall depth of the circuit by $\bigO(\log^2n)$. 
\hfill $\Box$
\end{proof}

\section{Conclusion and outlook}
The above discussion shows that the depth $\bigO(n^2)$ bound for discrete logarithm computations on ordinary binary elliptic curves can be improved exponentially: using parallelization at multiple levels, we can implement Shor's algorithm on complete binary Edwards curves in depth $\bigO(\log^2n)$. To show our result we introduced the first sublinear-depth circuits for ${\mathbb F}_{2^n}$-multiplication and ${\mathbb F}_{2^n}$-inversion, which may be of independent interest.
The price for the exponential reduction in depth is the introduction of ancillae---a parallelized ${\mathbb F}_{2^n}$-multiplier as discussed in Section~\ref{sec:fieldmult} adds $\bigO(n^2)$ additional qubits. Depending on the number of qubits available, for cryptographically significant parameters, say $n\ge 163$, trade-offs that put less emphasis on depth-optimization could be interesting. For instance, one could combine the parallelized double-and-add computation with a linear-depth field arithmetic. Even though not being depth-optimal, thousands of qubits could be saved in this way. One could also try to avoid the general point addition circuits in the parallelized double-and-add procedure (which involve many field multiplications) and instead rely on a sequential depth $\bigO(n\log n)$ solution with specialized point addition circuits as discussed in Section~\ref{sec:seq2a}. It is not clear at this point how the best trade-off for a large-scale discrete logarithm computation looks like.

It is interesting to compare our findings for bounding the depth of Shor's algorithm over an additive group on an ordinary binary elliptic curve to the case of factoring. In both cases the high-level structure---namely, phase estimation over an abelian group---of the algorithm is the same, however, the details on how to implement the arithmetic differ significantly. In the case of the elliptic curve arithmetic, the bottle neck are the 
addition formulae for points on the curve which involve finite field divisions. Our 
polylogarithmic depth implementation of Shor's algorithm may be compared with the findings of \cite{PS:2013} in which an $\bigO(\log^2n)$ upper bound on the depth for Shor's algorithm for integer factorization was given for a 2D nearest neighbor architecture and to \cite{Kutin2006}, where a circuit of depth $\bigO(n^2)$ was derived, albeit on a 1D nearest neighbor architecture.
The circuits presented in this paper assume an arbitrary coupling model and we leave it as an open problem whether they can be adapted to a 2D nearest neighbor architecture while maintaining the same upper bound on the circuit depth.

To enable parallelization, we make intensive use of `multi-fan-out CNOTs with $\ket{0}$-input', and one may hope for further depth improvements if multi-fan-out gates are provided. At the moment we do not know how to derive an asymptotic benefit of such gates that goes beyond a constant factor improvement. The remaining bottleneck for an asymptotic improvement in the ${\mathbb F}_{2^n}$-arithmetic seems the logarithmic depth for parity computations.
Finally, Amy et al.'s \cite{AMM13} observation that ${\mathbb F}_{2^n}$-multiplication can be realized in $T$-depth 2 suggests another natural direction for follow-up research: minimizing $T$-depth and the number of $T$-gates necessary to compute discrete logarithms on binary elliptic curves.

\paragraph{Acknowledgments.} We thank the reviewers for their constructive feedback and we thank Schloss Dagstuhl, Germany, for providing an excellent environment for part of this research through a \emph{Quantum Cryptanalysis} seminar. RS was supported by the Spanish \emph{Ministerio de Econom{\'\i}a y Competitividad} through the project grant MTM-2012-15167 and by NATO's Public Diplomacy Division in the framework of ``Science for Peace'', Project MD.SFPP 984520. This work was carried out while MR was with NEC Laboratories America, Princeton, NJ 08540, USA.


\begin{thebibliography}{10}

\bibitem{KaZa04}
Phillip Kaye and Christof Zalka.
\newblock {Optimized quantum implementation of elliptic curve arithmetic over
  binary fields}.
\newblock arXiv:quant-ph/0407095v1, July 2004.
\newblock Available at \url{http://arxiv.org/abs/quant-ph/0407095v1}.

\bibitem{MMCP09b}
Dmitri Maslov, Jimson Mathew, Donny Cheung, and Dhiraj~K. Pradhan.
\newblock {An $O(m^2)$-depth quantum algorithm for the elliptic curve discrete
  logarithm problem over GF$(2^m)$}.
\newblock {\em Quantum Information \& Computation}, 9(7):610--621, 2009.
\newblock For a preprint version see \cite{MMCP09}.

\bibitem{FIPS1864}
National Institute of Standards and Technology, Gaithersburg, MD 20899-8900.
\newblock {\em FIPS PUB 186-4. Federal Information Processing Standard
  Publication. Digital Signature Standard (DSS)}, July 2013.
\newblock Available at
  \url{http://nvlpubs.nist.gov/nistpubs/FIPS/NIST.FIPS.186-4.pdf}.

\bibitem{Sho97}
Peter~W. Shor.
\newblock {Polynomial-Time Algorithms for Prime Factorization and Discrete
  Logarithms on a Quantum Computer}.
\newblock {\em SIAM Journal on Computing}, 26(5):1484--1509, 1997.

\bibitem{ARS12b}
Brittanney Amento, Martin R{\"o}tteler, and Rainer Steinwandt.
\newblock {Efficient quantum circuits for binary elliptic curve arithmetic:
  reducing $T$-gate complexity}.
\newblock {\em Quantum Information \& Computation}, 13:631--644, July 2013.

\bibitem{ARS12}
Brittanney Amento, Martin R{\"o}tteler, and Rainer Steinwandt.
\newblock {Quantum binary field inversion: improved circuit depth via choice of
  basis representation}.
\newblock {\em Quantum Information \& Computation}, 13:116--134, January 2013.

\bibitem{CoFr06}
Henri Cohen and Gerhard Frey, editors.
\newblock {\em {Handbook of Elliptic and Hyperelliptic Curve Cryptography}}.
\newblock Discrete mathematics and its applications. Chapman \& Hall/CRC, 2006.

\bibitem{Sol98}
Jerome~A. Solinas.
\newblock {An Improved Algorithm for Arithmetic on a Family of Elliptic
  Curves}.
\newblock In Burton S.~Kaliski Jr., editor, {\em Advances in Cryptology --
  CRYPTO '97}, volume 1294 of {\em Lecture Notes in Computer Science}, pages
  357--371. Springer, 1997.

\bibitem{BLF08}
Daniel~J. Bernstein, Tanja Lange, and Reza~Rezaeian Farashahi.
\newblock {Binary Edwards Curves}.
\newblock In Elisabeth Oswald and Pankaj Rohatgi, editors, {\em Cryptographic
  Hardware and Embedded Systems -- CHES 2008}, volume 5154 of {\em Lecture
  Notes in Computer Science}, pages 244--265. International Association for
  Cryptologic Research, Springer, 2008.

\bibitem{Kitaev:97}
Aleksei~Yu. Kitaev.
\newblock Quantum computations: algorithms and error correction.
\newblock {\em Russian Math. Surveys}, 52(6):1191--1249, 1997.

\bibitem{GHM02}
Frederic Green, Steven Homer, Cristopher Moore, and Christopher Pollett.
\newblock {Counting, Fanout, and the Complexity of Quantum ACC}.
\newblock {\em Quantum Information \& Computation}, 2(1):35--65, 2002.

\bibitem{Mas88}
Edoardo~D. Mastrovito.
\newblock {VLSI designs for multiplication over finite fields $GF(2^m)$}.
\newblock In Teo Mora, editor, {\em Proceedings of the Sixth Symposium on
  Applied Algebra, Algebraic Algorithms and Error Correcting Codes}, volume 357
  of {\em Lecture Notes in Computer Science}, pages 297--309. Springer, 1988.

\bibitem{Mas91}
Edoardo~D. Mastrovito.
\newblock {\em {VLSI Architectures for Computation in Galois Fields}}.
\newblock PhD thesis, Link{\"o}ping University, Link{\"o}ping, Sweden, 1991.

\bibitem{MaHa04}
Arash Reyhani-Masoleh and M.~Anwar Hasan.
\newblock {Low Complexity Bit Parallel Architectures for Polynomial Basis
  Multiplication over $GF(2^m)$}.
\newblock {\em IEEE Transactions on Computers}, 53(8):945--959, 2004.

\bibitem{MMCP09}
Dmitri Maslov, Jimson Mathew, Donny Cheung, and Dhiraj~K. Pradhan.
\newblock {On the Design and Optimization of a Quantum Polynomial-Time Attack
  on Elliptic Curve Cryptography}.
\newblock arXiv:0710.1093v2, February 2009.
\newblock Available at \url{http://arxiv.org/abs/0710.1093v2}.

\bibitem{Str64}
Ernst~G. Straus.
\newblock {Addition chains of vectors (problem 5125)}.
\newblock {\em American Mathematical Monthly}, 70:806--808, 1964.

\bibitem{ElG84}
Taher ElGamal.
\newblock {A public key cryptosystem and a signature scheme based on discrete
  logarithms}.
\newblock In George~R. Blakley and David Chaum, editors, {\em Advances in
  Cryptology -- CRYPTO '84}, volume 196 of {\em Lecture Notes in Computer
  Science}, pages 10--18. Springer, 1985.

\bibitem{ItTs89}
Toshiya Itoh and Shigeo Tsujii.
\newblock {Structure of parallel multipliers for a class of fields $GF(2^m)$}.
\newblock {\em Information and Computation}, 83:21--40, 1989.

\bibitem{CW:2000}
Richard Cleve and John Watrous.
\newblock {Fast parallel circuits for the quantum Fourier transform}.
\newblock In {\em Proceedings of the 41st Annual Symposium on Foundations of
  Computer Science (FOCS'00)}, pages 526--536, 2000.

\bibitem{PS:2013}
Paul Pham and Krysta Svore.
\newblock {A 2D nearest-neighbor quantum architecture for factoring in
  polylogarithmic depth}.
\newblock {\em Quantum Information \& Computation}, 13(11\&12):937--962, 2013.

\bibitem{Kutin2006}
Samuel Kutin.
\newblock {Shor's algorithm on a nearest-neighbor machine}.
\newblock arXiv:quant-ph/0609001, September 2006.
\newblock Available at \url{http://arxiv.org/abs/quant-ph/0609001}.

\bibitem{AMM13}
Matthew Amy, Dmitri Maslov, and Michele Mosca.
\newblock {Polynomial-time T-depth Optimization of Clifford+T circuits via
  Matroid Partitioning}.
\newblock arXiv:quant-ph/1303.2042, March 2013.
\newblock Available at \url{http://arxiv.org/abs/1303.2042}.

\end{thebibliography}

\end{document}